\documentclass[envcountsame]{llncs}
\usepackage{amsmath,amssymb,amsbsy,amsfonts,array,latexsym,
               amsopn,amstext,amsxtra,euscript,amscd,commath,mathtools,physics,hyperref}
\usepackage{color,soul}
\usepackage[english]{babel}
\usepackage[noadjust]{cite}
\usepackage{algorithm}
\usepackage{algcompatible}
\usepackage{caption}
\usepackage{tikz}
\usetikzlibrary{quantikz}

\newcolumntype{P}[1]{>{\centering\arraybackslash}p{#1}}

\def \F {{\mathbb F}}

\def \Z {{\mathbb Z}}

\newtheorem{prop}{Proposition}
\newtheorem{defn}{Definition}

\def\cO{{\mathcal O}}

\def\F{{\mathbb F}}
\def\Z{{\mathbb Z}}

\def\F{{\mathbb F}}

\def\R{{\mathbb R}}
\def\Z{{\mathbb Z}}

\def\R{{\mathbb R}}
\def\C{{\mathbb C}}

\def\00{{\bf 0}}
\def\11{{\bf 1}}
\def\+{\oplus}

\def \F {{\mathbb F}}

\def \Z {{\mathbb Z}}

\newcommand{\bias}[1]{\mathbf{B}[#1]}
\newcommand{\mean}[1]{\mathbf{E}[#1]}

\def\wt{{\rm wt}}

\usepackage[parfill]{parskip}
\begin{document}

\title{\bf  A quantum algorithm to estimate the closeness to the Strict Avalanche criterion in Boolean functions}

\author{ C.~A.~Jothishwaran\inst{1}, Abhishek Chakraborty\inst{3}, Vishvendra Singh Poonia\inst{1}\\ Pantelimon St\u anic\u a\inst{3},
Sugata~Gangopadhyay\inst{2}}
\institute{Department of Electronics and Communication Engineering\\
\and
Department of Computer Science and Engineering,\\
Indian Institute of Technology Roorkee, Roorkee 247667, INDIA\\ \email{\{jc\_a, vishvendra\}@ece.iitr.ac.in, sugata.gangopadhyay@cs.iitr.ac.in}\\
\and Department of Physics and Astronomy \\
University of Rochester, Rochester, NY 14627--0171, USA \\ \email{achakra9@ur.rochester.edu} \\
\and Department of Applied Mathematics  \\
Naval Postgraduate School, Monterey, CA 93943--5216, USA \\ \email{pstanica@nps.edu}}

\maketitle

\abstract{We propose a quantum algorithm (in the form of a quantum oracle) that estimates the closeness of a given Boolean function to one that satisfies the ``strict avalanche criterion'' (SAC). This algorithm requires $n$ queries of the Boolean function oracle, where $n$ is the number fo input variables, this is fewer than the queries required by the classical algorithm to perform the same task. We compare our approach with other quantum algorithms that may be used for estimating the closeness to SAC and it is shown our algorithm verifies SAC with the fewest possible calls to quantum oracle and requires the fewest samples to for a given confidence bound.}


\section{Introduction}

An $n$-variable Boolean function $F$ is a function from $\F_2^n$ to $\F_2$. 
The set of all such functions is denoted by  $\mathfrak B_n$. We associate to each function 
$F \in \mathfrak B_n$ its character form $f : \F_2^n \rightarrow \R$, 
defined by $f(x) = (-1)^{F(x)}$, for all $x \in \F_2^n$ (we use capitals letters for classical Boolean functions and lower cases for their signatures). In this article, 
abusing notation, we refer to the character forms $f$ as Boolean functions and 
go to the extent of writing $f \in \mathfrak B_n$, whenever $F \in \mathfrak B_n$, surely, 
if there is no danger of confusion. For any $x, y \in \F_2^n$, the inner product is
$x \cdot y = \sum_{i = 1}^n x_i y_i$, where the sum is over $\F_2$. The (Hamming) 
weight of a vector $u=(u_1,\ldots,u_n) \in \F_2^n$ is $\wt (u) = \sum_{i = 1}^n u_i$, 
where the sum is over~$\Z$. The weight of a Boolean function $F \in \mathfrak B_n$, 
or equivalently $f \in \mathfrak B_n$ is the cardinality
$\wt (F) = \abs{ \{x \in \F_2^n: F(x) \neq 0 \}}$, or equivalently
$\wt (f) = \abs{ \{x \in \F_2^n: f(x) \neq 1 \}}$. 
We define the Fourier coefficient of $f$ at $u \in \F_2^n$ by 
$\widehat f(u) = 2^{-n} \sum_{x \in \F_2^n} f(x)(-1)^{u \cdot x}$.
Recall the well known Parseval's identity $\sum_{u \in \F_2^n}\widehat f(u)^2 = 1$.
The derivative of $f \in \mathfrak B_n$ at $c \in \F_2^n$ is the function
$\Delta_c f(x) = f(x)f(x+c)$ (or, equivalently, $\Delta_c F(x) = F(x+c)+F(x)$),  for all  $x,c \in \F_2^n$.

\section{Strict Avalanche Criterion}

A Boolean function $f \in \mathfrak{B}_n$ satisfies the strict avalanche criterion (SAC) if the probability of the function changing its value when a single input value is flipped is exactly $0.5$ that is, the derivative $\Delta_cF(x)$ is a balanced function for all $c \in \F_2^n$ such that $\wt(c) = 1$. 
We refer to Budaghyan~\cite{Bud14}, Carlet~\cite{CH1}, and Cusick and St\u anic\u a~\cite{Pante} for detailed discussions on SAC and other cryptographic properties of Boolean functions. 
The Fourier transform of the function $\widehat{f}(w)$ satisfies the following relation
\begin{equation}\label{SAC_criterion}
\sum_{w \in \F_2^n} \; (-1)^{w \cdot c} \; \widehat{f}(w)^2 \,=\, 0, \mbox{ for all } c \in \F_2^n 
\mbox{ such that } \wt(c) = 1.
\end{equation}

As there are $n$ possible strings $c$ of weight~$1$, the above expression represents $n$ different relations simultaneously satisfied by $\widehat{f}(w)$. These relations can also be written in terms of  $w_i$, the $i^{\text{th}}$ bit of $w$ as follows:
\begin{equation}\label{SAC_criterion_2}
\begin{split}
&\; \sum_{w \in \F_2^n} \; (-1)^{w_i} \; \widehat{f}(w)^2 \,=\, 0, \forall\; i \in [n] \,:\, [n] \equiv \{1,2,\ldots,n\}, \\
\mbox{i.e., } &\; \sum_{w \in \F_2^n \vert w_i = 0} \widehat{f}(w)^2 \,=\, \sum_{w \in \F_2^n \vert w_i = 1} \widehat{f}(w)^2 \,=\, \frac{1}{2},\, \forall\; i \in [n].
\end{split}
\end{equation} 
The last relation is obtained by using Parseval's identity. SAC can also be defined in terms of the autocorrelation coefficients of the Boolean function $f(x)$ given by:

\begin{equation}\label{autocorrelation}
\breve{f}(a) = \sum_{x \in \F_2^n} \, f(x)f(x+a) \;\;\forall\; a \in \F_2^n
\end{equation}

By the above definition, any Boolean function $f$ satisfying SAC will have
\begin{equation}\label{autosac}
\breve{f}(a) = 0 \;\forall\; a \in \F_2^n : \text{wt}(a) = 1
\end{equation} 

Finally, the autocorrelation coefficients are related to the Fourier coefficient through the inverse transform $f(x) = \sum_{u \in \F_2^n} \widehat{f}(u) (-1)^{u \cdot x}$ as follows

\begin{equation}\label{auto-fourier}
\breve{f}(a) \,=\, 2^n \,\sum_{u \in \F_2^n} \widehat{f}(u)^2 (-1)^{u \cdot a}
\end{equation}

Therefore, the autocorrelation condition for the criterion is equivalent to the first result in \autoref{SAC_criterion_2}. A classical algorithm for estimating the closeness to SAC functions is given in the following section. 

\section{Classically estimating closeness the to SAC functions}\label{class_alg}

Since the strict avalanche criterion (SAC) for a function $f \in \mathfrak{B}_n$ is defined in terms of the balanced nature of the derivatives, $\Delta_a f$ corresponding to all $c \in \F_2^n \,:\, \text{wt}(a) = 1$. The closeness of a given function to a SAC function can be estimated by estimating the bias of the relevant derivative functions. The bias of $g \in \mathfrak{B}_n$ can be in terms of the expectation value $\mean{g}$ under uniform sampling of the inputs. The bias can also be defined in terms of the weight of a Boolean function. The following are the definitions of the bias $\bias{g})$:

\begin{equation}
\begin{split}
\bias{g} &\,=\, 2^n\,\mean{g} \,=\,2^n\,\qty(\Pr[g(x) = 1] - \Pr[g(x) = -1]) \\ 
&\,=\, \qty[2^n - 2 \cdot \text{wt}(g)]
\end{split}
\end{equation}\label{bias}

In general $-2^n \leq \bias{g} \leq 2^n$, for the case of a balanced function, $\bias{g} = 0$. Algorithm~\ref{algo-class} can be used to test the closeness of a given function $f \in \mathfrak{B}_n$ to a SAC function.

{
\center
\begin{algorithm}[htb]
\centering
\begin{algorithmic}[1]
\STATEx{Input: Oracle access to the Boolean function $f$ }
\STATE{Choose $a \in \F_2^n$ such that $\wt(a) = 1$. }
\STATE{Take $m$ uniformly random samples of $x \in \F_2^n$ listed as $\qty{x^{(1)},x^{(2)},\ldots,x^{(m)}}$.} 
\STATE{Calculate the sample mean of the function $f(x)f(x+a)$, $\epsilon_a$ given by:
\begin{equation*}
s \,=\, \frac{1}{m}\,\sum_{i = 1}^{m}\, f(x^{(i)})f(x^{(i)}+a)
\end{equation*}}
\STATE{Estimate the bias as, $\tilde{\epsilon}_a \,=\, 2^n s $}.
\STATE{Repeat for all allowed values of $a$.} 
\STATE{Evaluate the quantity $$\tilde{\epsilon} \,=\, \frac{1}{2}\sum_{\substack{a \in \F_2^n \\
\wt(a) = 1}} \epsilon_a $$}
\STATE{Return $\epsilon$ as the estimated distance to a SAC function}
\end{algorithmic}
\caption{Estimating the closeness to a SAC function.}
\label{algo-class}
\end{algorithm}
}

\begin{theorem}
If the bias of the derivative function $f(x)f(x+a)$ for all $\wt{a} = 1$ is given by $\epsilon_a$, then $f$ is at most $\epsilon$ distant from a SAC function, where $\epsilon$ is defined as:  

\begin{equation}
\epsilon \,=\, \frac{1}{2}\sum_{\substack{a \in \F_2^n \\
\wt(a) = 1}} \epsilon_a 
\end{equation}\label{closeness}
\end{theorem}\label{class-test}

\begin{proof}
It can be readily noted that the bias, $\bias{f(x)f(x+a)}$ is numerically equal to the autocorrelation $\breve{f}(a)$. In the case of a SAC function $g$, $\breve{g}(a) = 0$ for all relevant values as mentioned in \autoref{autosac}.

This implies that for each allowed value of $a$, the value of the function $f(x + a)$ must be flipped at exactly $\epsilon_a/2$ points for the bias of the derivative $\bias{f(x)f(x+a)}$ to vanish. When this process can be simultaneously done for all $n$ allowed values of $a$, the given Boolean function can be converted to a SA function. 

Since the number of value inversion of the function is given by $\epsilon_a/2$ for each allowed $a$, the maximum number of points where the value of $f$ may need be changed is given by $\epsilon$ where,

$$\epsilon \,=\, \frac{1}{2}\sum_{\substack{a \in \F_2^n \\
\wt(a) = 1}} \epsilon_a $$

Therefore, there exists a Boolean function $g$ that satisfies the strict avalanche criterion, that is at most $\epsilon$-distant from the given Boolean function $f$. \qed   
\end{proof}

Theorem~\ref{class-test} establishes that Algorithm~\ref{algo-class} can be used to estimate the closeness of a given Boolean function to some SA function. However, the issue of the accuracy of the bias estimates must be addressed and this will be done in \autoref{complexity}. 

\section{Quantum information: definitions and notation} \label{quant_def}

In this section, we will introduce some notation that we use throughout the paper. For an introduction to quantum computing, we refer to Rieffel and Polak \cite{Rieffel}, or Nielsen and Chuang~\cite{NC10}.
The fundamental unit of quantum information is called a qubit. The states of a qubit is 
denoted by $\ket{\psi} = a \ket{0} + b \ket{1}$ where $a, b \in \C$ such that $\abs{a}^2 + \abs{b}^2 =1$.   
If we measure the qubit $\ket{\psi}$ using the standard basis $\{\ket{0}, \ket{1}\}$
the probabilities of observing $\ket{0}$ and $\ket{1}$ are $\abs{a}^2$ and $\abs{b}^2$, respectively. 

In the following, we will use the conventional notation $\ket{a} \ket{b}:=\ket{a} \otimes \ket{b}$,  
or $\ket{ a  b}:=\ket{a} \otimes \ket{b}$. A state on $n$ qubits can be represented as a 
$\C$-linear combination of the vectors of the standard basis 
$\ket\psi=\sum_{x\in\F_2^n}a_x\ket x$, where $a_x\in\C,$ for all $x\in \F_2^n$, and $\sum_{x\in\F_2^n}|a_x|^2=1$; the set of vectors ${\ket x}$ forms a basis for the $n$ 
qubit states and is referred to as the computational basis.
Let $\ket{0_n}$ be the quantum state associated with the zero vector in $\F_2^n$. 
The vectors $\ket + = \frac{\ket 0 + \ket 1}{\sqrt 2}$ and $\ket - = \frac{\ket 0 - \ket 1}{\sqrt 2}$
 define the Hadamard basis for single qubit states. 
 
Any Boolean function $F \in \mathfrak B_n$ can be implemented as a bit oracle implementation 
$U_F$, so that:
\begin{equation}
\label{bit-oracle}
\ket x \ket \varepsilon \xrightarrow{U_F} \ket x \ket {\varepsilon + F(x)}.
\end{equation}
Here, $x \in \F_2^n$ and $ \varepsilon \in \F_2$.  
%The bit oracle $U_F$ maps vectors belonging to the computational basis of the $n+1$ qubit states to other vectors in the same basis. $U_F$ is therefore a unitary transformation. The $n$ qubit state vector $\ket x$ is the input to the oracle and the state $\ket \varepsilon$ is the target qubit. The target qubit shall be referred to as $\ket{\varepsilon}$ and a measurement of the target qubit with respect to the computational basis can be labelled by $\varepsilon$ and takes values $\{ 0, 1 \}$.
If the target qubit for $U_F$ is $\ket -$, then $$\ket x \ket - \xrightarrow{U_F} (-1)^{F(x)}\ket x \ket -.$$ 
This gives an alternative implementation of a Boolean function oracle known as the phase oracle implementation of the function~$F$. 
% and is represented as $\ket x \xrightarrow{U_F} (-1)^{F(x)} \ket x$ with the understanding that there is an additional target qubit in the $\ket{-}$ state that remains unchanged. 
The number of input qubits to a quantum oracle of Boolean function in $\mathfrak{B}_n$ is $n+1$, any internal ancillary qubits used are not required for the analysis of the algorithm and are not included in this count. 

If the $i^{\text{th}}$ input qubit acts as the control qubit and the target qubit is the same as the target of the oracle, the gate is represented by $CZ^{i}$. In the phase oracle representation, the target qubit is in the $\ket{-}$ state the action of $CZ^{i}$ is as follows,
\begin{equation}\label{CZ_act}
\ket{x}\ket{-} \,\xrightarrow{CZ^{i}}\, \ket{x}\mqty(\frac{\ket{0} - (-1)^{x_i} \ket{1}}{\sqrt{2}}).
\end{equation}  
This implies if the $x_i =1$ the target qubit turns from $\ket{-}$ to $\ket{+}$.

Let $I = \begin{pmatrix}1&0\\0&1 \end{pmatrix} $ be the $2\times 2$ identity matrix, and $H = \frac{1}{\sqrt{2}} \begin{pmatrix*}[r] 1 & 1 \\ 1 & -1 \end{pmatrix*}$ be the $2\times 2$ Hadamard matrix. The tensor product of matrices is denoted by $\otimes$. The matrix $H_n$ is recursively defined as: 
\begin{equation}
\label{def-hadamard}
\begin{split}
H_2 &= H \otimes H, \\
H_n &= H \otimes H_{n-1}, \mbox{ for all } n \geq 3.
\end{split}
\end{equation}
Note that, for $x\in\F_2^n$, $H_n\ket x=2^{-\frac{n}{2}}\sum_{x'\in\F_2^n}(-1)^{x\cdot x'}\ket{x'}$.

\section{Quantum algorithm to verify the strict avalanche criterion}

The SAC can be verified for a function $F \in \mathfrak B_n$ through the following quantum algorithm, henceforth referred to as the QSAC algorithm. The initial state of the $n+1$ qubits is $\ket{0_n}\ket{0}$. In the following expressions, the symbol ``$\circ$'' is used to represent matrix multiplication:

\allowdisplaybreaks
\begin{eqnarray}
\label{main-algo-1}
\ket{0_n}\ket{0} &&\xrightarrow{H_n\otimes(Z \circ H)} 2^{-n/2}\sum_{x \in \F_2^n} \; \ket{x}\ket{-} \nonumber\\
&&\xrightarrow{U_F \otimes I} 2^{-n/2}\sum_{x \in \F_2^n} \; f(x)\ket{x}\ket{-} \nonumber\\
&&\xrightarrow{H_n\otimes I} 2^{-n}\sum_{x \in \F_2^n} \,f(x)\, \sum_{w \in \F_2^n} \; (-1)^{x \cdot w}\ket{w}\ket{-} \nonumber\\
&&\equiv 2^{-n}\sum_{w \in \F_2^n}  \sum_{x \in \F_2^n} \; f(x)(-1)^{x \cdot w}\ket{w}\ket{-} \,=\, \sum_{w \in \F_2^n} \; \widehat{f}(w)\ket{w}\ket{-}\\
&&\xrightarrow{CZ^{i}} \sum_{w \in \F_2^n} \; \widehat{f}(w)\ket{w}\qty(\frac{\ket{0} - \,(-1)^{w_i}\ket{1}}{\sqrt{2}}) \nonumber\\
&&\equiv \sum_{w \in \F_2^n \vert w_i = 1} \; \widehat{f}(w)\ket{w}\ket{+} \;+ \sum_{w \in \F_2^n \vert w_i = 0} \; \widehat{f}(w)\ket{w}\ket{-} \nonumber\\
&&\xrightarrow{I_n \otimes H} \sum_{w \in \F_2^n \vert w_i = 1} \; \widehat{f}(w)\ket{w}\ket{0} \;+ \sum_{w \in \F_2^n \vert w_i = 0} \; \widehat{f}(w)\ket{w}\ket{1}. \nonumber
\end{eqnarray}

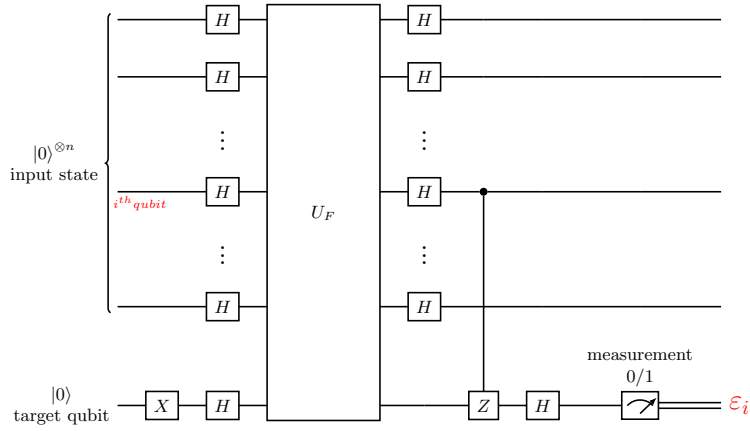
\begin{figure}[htb]
\begin{center}
\scalebox{0.75}{
\begin{quantikz}
& \lstick[wires=6]{$\ket{0}^{\otimes n}$ \\ \small{input state}} & \qw  & \gate{H} & \gate[wires=7, nwires={3,5}][2cm]{U_F} & \gate{H} & \qw & \qw & \qw & \qw  \\
&  & \qw  & \gate{H} &  & \gate{H} & \qw & \qw & \qw & \qw  \\
&  &  & \Large{\vdots} &  & \Large{\vdots} &  &  &  &   \\
&  & \qw{\textcolor{red}{i^{th} qubit}}  & \gate{H} &  & \gate{H} & \ctrl{3} & \qw & \qw & \qw  \\
&  &  & \Large{\vdots} &  & \Large{\vdots} &  &  &  &   \\
&  & \qw  & \gate{H} &  & \gate{H} & \qw & \qw & \qw & \qw  \\
& \lstick{$\ket{0}$ \\ \small{target qubit}} & \gate{X} & \gate{H} &  & \qw & \gate{Z} & \gate{H} & \meter{measurement \\ $0/1$} & \cw \rstick{\Large{\textcolor{red}{$\varepsilon_i$}}}  \\
\end{quantikz}}
\caption{\footnotesize{Circuit diagram for a single iteration of the algorithm. The outcome, $\varepsilon_i$ takes values 0 or 1. If the function satisfies SAC, the mean value of the outcome is 0.5.}}
\label{qsac_circuit}
\end{center}
\end{figure}

The probability that the outcome of a standard measurement on the target qubit $\ket{\varepsilon}$ will yield $\ket{0}$  or $\ket{1}$ is  given by:
\begin{equation}\label{sac_result}
\begin{split}
\Pr[\varepsilon = 0] &\,=\, \sum_{w \in \F_2^n \vert w_i = 1} \widehat{f}(w)^2, \\
\Pr[\varepsilon = 1] &\,=\, \sum_{w \in \F_2^n \vert w_i = 0} \widehat{f}(w)^2,
\end{split}
\end{equation}
for a function satisfying SAC these two probabilities are $\frac{1}{2}$ and the expectation value $\ev{\varepsilon}$ is also $\frac{1}{2}$. it should be noted that the same algorithm must be iteratively repeated $n$ times with a different gate  $$ CZ^i \forall\; i \in [n], $$ and yield the same value for $\ev{\varepsilon}$ for the target qubit. The circuit representation of an iteration of the algorithm is given in \autoref{qsac_circuit}. 

It can be seen that this algorithm is a quantum implementation of the expression in \autoref{SAC_criterion_2}. The quantity $\ev{\varepsilon_i}$ obtained from the $i^{\text{th}}$ iteration of the quantum algorithm is related to results from the iteration of the classical algorithm where $a_i = 1$ using \autoref{auto-fourier} and can be derived as follows:

\begin{equation}\label{class-quant}
\begin{split}
\breve{f}(a) &\,=\, 2^n \,\qty(\sum_{\substack{w \in \F_2^n  \\ \wt(a) = 1 \vert a_i = 1}} \widehat{f}(w)^2 (-1)^{w_i})\\
&\,=\, 2^n \,\qty(\Pr[\varepsilon_i = 0] - \Pr[\varepsilon_i = 1]) \\
&\,=\, 2^n \,\qty(1 - 2\ev{\varepsilon_i}) \\
\Rightarrow\; \epsilon_a &\,=\, 2^n\,\qty(1 - 2\ev{\varepsilon_i}) 
\end{split}
\end{equation}

The above relation can be used to compute the distance of $f$ to a SAC function in a manner identical to Algorithm~\ref{algo-class}. Therefore the quantum algorithm to estimate the closeness of a given function to a SAC function is given as follows:

{
\center
\begin{algorithm}[htb]
\centering
\begin{algorithmic}[1]
\STATEx{Input: Quantum computer with a minimum $n+1$ qubits, Bit oracle implementation to the Boolean function $f$}
\STATE{Initial state: $\ket{0_n}\ket{0}$}
\STATE{Apply the transformations described in \autoref{main-algo-1}/\autoref{qsac_circuit}}
\STATE{Repeat above steps multiple times and measure in order to estimate $\ev{\varepsilon_i}$}.
\STATE{Repeat for all values of $i$.}
\STATE{Evaluate the quantity $$\epsilon \,=\,\sum_{\substack{i \in [n]}} 2^{n-1} \,(1-\ev{\varepsilon_i}) $$}
\STATE{Return $\epsilon$ as the estimated distance to a SAC function} 
\end{algorithmic}
\caption{Quantum algorithm for estimating the closeness to a SAC function.}
\label{algo-quant}
\end{algorithm}
}
  
\section{Alternative quantum algorithms for verifying SAC}

There are several quantum algorithms that can analyse different properties of Boolean functions. This section considers certain suitable candidate algorithms and modifies them to verify SAC. These algorithms are considered as alternate approaches to verifying SAC. It should be noted that none of the candidate algorithms considered in this section have been used to explicitly verify SAC in their respective source materials. 

Additionally, only the quantum circuits corresponding to the algorithm are mentioned here. In all of the following cases, the estimated probability/expectation values can be converted to distance measures to a SAC function, but the proof for them is not explicitly given. 

\subsection{The direct approach}

This algorithm utilizes the fact that $\Delta_c F(x)$ is balanced for all $c \in \F_2^n$ such that $\text{wt}(c) = 1$. This implies the Fourier coefficient of these derivatives at the point $w = 0_n$ is zero. The character form of the derivative $\Delta_c F(x)$ is given by $f'_c(x) = g_c(x)$, the Fourier coefficient can be represented as $\widehat{g}_c(w)$. Therefore,  if $f$ satisfies SAC, then $\widehat{g_c}(0_n) = 0 \; \forall \; c \in \F_2^n \,:\, \text{wt}(c) = 1$. 

The direct algorithm is derived form the Deutsch-Jozsa \cite{Deutsch} algorithm. If the quantum oracle $U_f$ exists, then the oracle for any of the derivatives $U_{f'_c}$ can be constructed using two such oracles and any two $X$ gates applied to the $i^{\text{th}}$ qubit. The resultant state will not contain the vector $\ket{0_n}$ and therefore measurement of this resultant state will never yield the outcome $0_n$. SAC can be verified by repeating this algorithm for each derivative $f'_c(x)$. The circuit diagram for  a single iteration of this algorithm is given in \autoref{direct_algo}.

\begin{figure}[htb]
\centering
\scalebox{0.75}{
\begin{quantikz}
& \lstick[wires=6]{$\ket{0}^{\otimes n}$ \\ \small{input state}} & \qw  & \gate{H} & \gate[wires=7, nwires={3,5}][2cm]{U_f} & \qw & \gate[wires=7, nwires={3,5}][2cm]{U_f} & \qw & \gate{H} & \meter{$0/1$} & \cw  \\
&  & \qw  & \gate{H} &  & \qw &  & \qw & \gate{H} & \meter{$0/1$} & \cw  \\
&  &  & \Large{\vdots} &  &  &  &  & \Large{\vdots} &  &  \\
&  & \qw{\textcolor{red}{i^{th} qubit}}  & \gate{H} &  & \gate{X} &  & \gate{X} & \gate{H} & \meter{$0/1$} & \cw  \\
&  &  & \Large{\vdots} &  &  &  &  & \Large{\vdots} &  &  \\
&  & \qw  & \gate{H} &  & \qw &  & \qw & \gate{H} & \meter{$0/1$} & \cw  \\
& \lstick{$\ket{0}$ \\ \small{target qubit}} & \gate{X} & \gate{H} &  & \qw &  & \qw & \qw & \qw & \\
\end{quantikz}}
\caption{\footnotesize{An iteration of the direct algorithm. If $f(x)$ satisfies SAC, the measurement outcome will never be all $0'$s}.}
\label{direct_algo}
\end{figure}
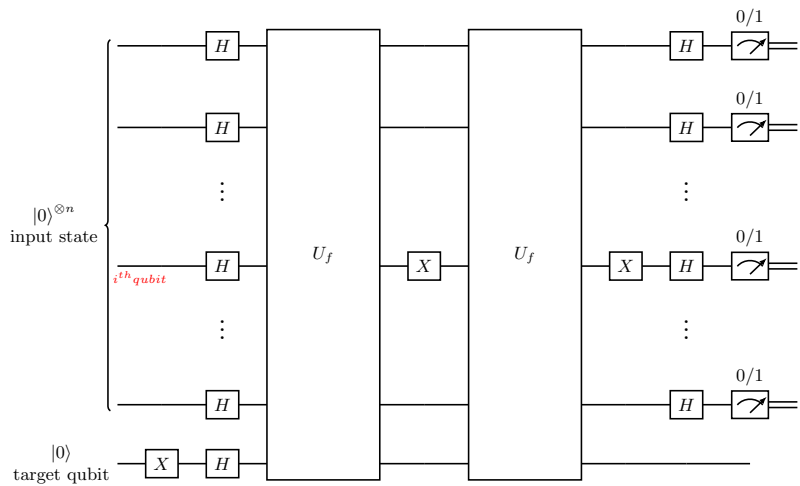

\subsection{The autocorrelation algorithm}

As shown in \autoref{autosac}, SAC can be defined explicitly in terms of $\breve{f}(a)$, the autocorrelation coefficient. for a given Boolean function $F \in \mathfrak B_n$, there are $2^n$ autocorrelation coefficients. Verifying SAC however requires only the $n$ coefficients corresponding to strings of weight $1$. Bera et al.~\cite{BeraMT19} deal with algorithms pertaining to the autocorrelation coefficients of Boolean functions. Two algorithms capable of verifying SAC can be derived form their work.

The first algorithm can be obtained by modifying the ``Higher order Deutsch-Jozsa'' algorithm $HoDJ^1_n$ (defined in Theorem 1 of~\cite{BeraMT19}) so as to include only shifts corresponding to strings of weight $1$. This would require the circuit corresponding to each of these strings is implemented iteratively for each string, similar to the QSAC algorithm. This algorithm turns out to be exactly identical to the direct algorithm previously explained.

The second algorithm can be obtained by modifying the autocorrelation estimator algorithm (Algorithm 2 in \cite{BeraMT19}). The quantum circuit for this algorithm is as shown in \autoref{auto_algo}. This circuit can be used to verify SAC by converting the register $R_1$ to a single qubit register and replacing the $CNOT$ gates shown here with $CNOT^i$ which is defined analogous to the $CZ^i$ gate defined in \autoref{CZ_act}. If the function considered satisfies SAC, then the output qubit $R_2$ will be 0 and 1 with probability $\frac{1}{2}$. The algorithm contains $n$ iterations of the circuit corresponding to the different $CNOT^i$ gates.  

\begin{figure}[htb]
\centering
\includegraphics[scale=0.75]{"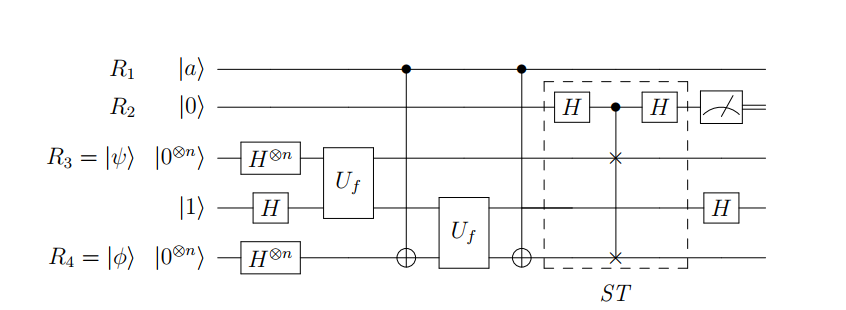"}
\caption{\footnotesize{The autocorrelation estimator algorithm due to Bera et al. The part of the circuit marked as ``ST'' stands for the Swap test. Figure source: \cite{BeraMT19}}}
\label{auto_algo}
\end{figure}

It should be noted that the autocorrelation algorithm has the same form of  output as the QSAC algorithm. However, this algorithm requires over double the amount of qubits on its initial state. The swap test performed before the measurement requires $n$ controlled-SWAP gates and each controlled-SWAP gate is implemented using three Toffoli gates and this introduces further complications in the implementation of the algorithm. These will be discussed during the complexity analysis of the algorithms.          

\subsection{The Forrelation algorithm}

The Forrelation problem defined by Aaronson et al.~\cite{Aaronson} can be used to analyse a number of Boolean function properties. Aaronson et al.~\cite{Aaronson}  also gave a quantum algorithm that evaluates the Forrelation between Boolean functions given their quantum oracles. The $k-$fold  Forrelation ($\Phi$) between $k$ Boolean functions $f_1, f_2 \dots f_k$ is given by
\begin{equation}\label{forrelation}
\Phi_{f_1, f_2 \dots f_k} \,=\,  \frac{1}{2^{(k+1)n/2}}\sum_{x_1,x_2 \dots x_k \in \F_2^n} \,  f_1(x_1) (-1)^{x_1 \cdot x_2} f_2(x_2) \dots (-1)^{x_{k-1} \cdot x_k} f_k{x_k}.
\end{equation}
The quantum algorithm evaluating this quantity calls each of the Boolean function oracles once.

Datta et al.~\cite[Lemma 3]{Datta} used the quantum algorithm for 3-fold Forrelation to verify the $m-$resilience of a general Boolean function $g(x)$. This algorithm evaluates 3-fold Forrelation $\Phi_{g,h,g}$ where $h(x)$ is a symmetric Boolean function such that $h(x) = 1 \;\forall\; x \in \F_2^n \,:\, wt(x) > m$. If $g(x)$ is $m-$resilient, the final state of the algorithm is always $\ket{0_n}$.

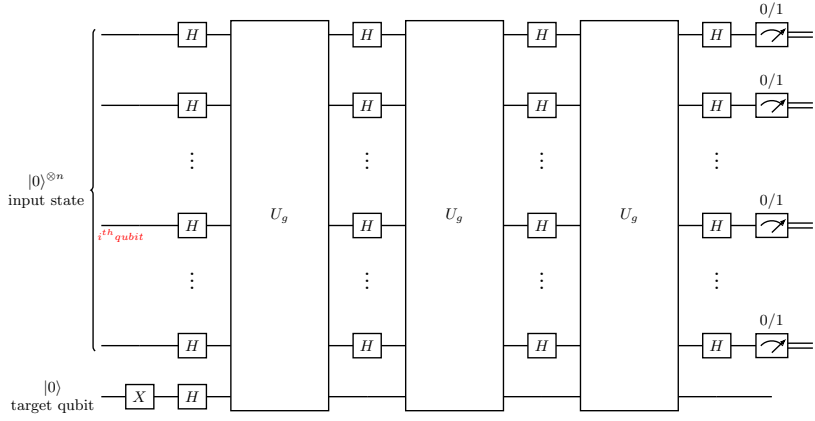
\begin{figure}[htb]
\centering
\scalebox{0.65}{
\begin{quantikz}
& \lstick[wires=6]{$\ket{0}^{\otimes n}$ \\ \small{input state}} & \qw  & \gate{H} & \gate[wires=7, nwires={3,5}][2cm]{U_g} & \gate{H} & \gate[wires=7, nwires={3,5}][2cm]{U_g} & \gate{H} & \gate[wires=7, nwires={3,5}][2cm]{U_g} & \gate{H} & \meter{$0/1$} & \cw  \\
&  & \qw  & \gate{H} &  & \gate{H} &  & \gate{H} &  & \gate{H} & \meter{$0/1$} & \cw  \\
&  &  & \Large{\vdots} &  & \Large{\vdots} &  & \Large{\vdots} &  & \Large{\vdots} &  &  \\
&  & \qw{\textcolor{red}{i^{th} qubit}}  & \gate{H} &  & \gate{H} &  & \gate{H} &  & \gate{H} & \meter{$0/1$} & \cw  \\
&  &  & \Large{\vdots} &  & \Large{\vdots} &  & \Large{\vdots} &  & \Large{\vdots} &  &  \\
&  & \qw  & \gate{H} &  & \gate{H} &  & \gate{H} &  & \gate{H} & \meter{$0/1$} & \cw  \\
& \lstick{$\ket{0}$ \\ \small{target qubit}} & \gate{X} & \gate{H} &  & \qw &  & \qw &  & \qw & \qw &  \\
\end{quantikz}}
\caption{\footnotesize{The Forrelation algorithm, if $f(x)$ satisfies SAC, all measurements yield the result `$0$'}. $g(z) = f'_c(x)$, and $h(x)$ is the quantum oracle for the $n$-input NOR gate.}
\label{forrel_algo}
\end{figure}

The balanced nature of $f'_c(x)$ also implies that it can be considered to be $0-$resilient. This can be verified setting $g = f'_c$ and using the appropriate $h(x)$ which in this case is just the $n$ variable NOR gate. The circuit diagram for an iteration of this algorithm is given in \autoref{forrel_algo}.

SAC can therefore be verified by running the algorithm for each of the derivative function. It can be seen that the resultant measurement of the Forrelation based algorithm is the complement of the approach based on the Deutsch-Josza algorithm.

\section{Complexity analysis of the algorithms}\label{complexity}

Since quantum algorithms yield probabilistic results, it also becomes necessary to sample the outcome of a particular algorithm several times to obtain the necessary statistics. Therefore the time complexity and the sampling complexity both become relevant. 

In general, a quantum algorithm may contain one or more quantum circuits and these circuits are implemented several times to obtain statistical results. The quantum circuit can also contain pre-defined objects like oracles whose internal structure is not considered in the complexity analysis. The complexity analysis of quantum algorithms shall use the following definitions.

\begin{defn}[Time Complexity]\label{time}
The time complexity of a quantum algorithm concerns the time required to implement/execute quantum circuit corresponding to the algorithm is represented by a pair of numbers $(m,k)$. Here, $m$ is the number of times pre-defined objects such as oracles are implemented (called) in the quantum circuit and $k$ is the number of gate operations (outside of the oracles) used in the circuit.
\end{defn}     

\begin{defn}[Sampling Complexity]\label{sample}
The sample complexity is given by a single number $q$, which is the minimum number of times a particular quantum circuit needs to be implemented and measured (queried) in order to obtain the desired results. This number is typically a written as a function of error/confidence bounds.  
\end{defn} 

In addition to this there is also the need to identify the minimum number of qubits that are required for the implementation of the quantum circuit. This measure is equivalent to the classical notion of space complexity of algorithms.  

In case any of these values scale with the input size of the quantum algorithm, these numbers can be replaced by appropriate asymptotic forms. The time complexity can be directly calculated by inspecting the quantum algorithm. sample complexity on the other hand can be estimated using several methods. In the discussion presented here, the sample complexity is estimated using the Hoeffding's inequality. 

\begin{prop}[Hoeffding's inequality]\label{inequality}
Given a random variable $X$ in the range $[a,b]$ and $m$ independent samples of $X$ given by $X_1, X_2, \ldots , X_m$. The expectation value is given by $E[x] = \mu$ and the mean of the samples is given by $$\overline{X} = \frac{1}{m} \, \sum_{i = 1}^m \, X_i$$ Hoeffding's inequality gives a bound on the value of $\overline{X}$  as $$\Pr[ \overline{X} - t \leq \mu \leq \overline{X} + t ] \geq 1 - 2 \, \exp(-\frac{2 m t^2}{(b - a)} ).$$   
\end{prop}  

\subsection{Complexity of the Classical algorithm}

In the case of the classical algorithm, The distance $(\epsilon)$ estimation is done through the estimated values of the bias of the derivative functions $\tilde{\epsilon}_a$ and these quantities are the sample means of their respective sampling experiments, the accuracy of these estimates can be determined using the Hoeffding's inequality with $b = 2^n$ and $a = -2^n$. This can be expressed as follows:

\begin{equation}\label{class_sample complexity}
\Pr[\tilde{\epsilon}_a - t \leq \epsilon_a \leq \tilde{\epsilon}_a + t] \,\geq\, 1 - 2\,\exp(-\frac{2mt^2}{2^{n+1}})\;;\; t > 0 \;\text{and}\; m \in \Z^+
\end{equation}

Therefore, for a confidence interval of $\delta$ that the sample mean is within an error margin of $t$ to the actual mean, the number of samples $m$ required is given by $$m \,=\, \frac{2^n}{t^2}\log(\frac{2}{\delta})$$

It can be seen that in the classical algorithm, for each sample value the Boolean function oracle is called twice, therefore the query complexity is $2n$.   

\subsection{Complexity of the QSAC algorithm}

Using definition~\ref{inequality}, the time complexity for a single iteration of the QSAC algorithm is given by $(1, 2n+4)$. This shows that the number of additional gates required per iteration is $\cO(n)$.

Let $\ket{\varepsilon^{(i)}} \;\forall\;  i \in \Z^+ :  1 \leq i \leq n $ denote the state of the target qubit after the $i^{\text{th}}$ iteration. Therefore, the probability distribution of the associated random variable $\varepsilon^{(i)}$ is the same as given in \eqref{sac_result}. We prove the following theorem.

\begin{theorem}\label{sample_complexity}
Let the probability $\Pr[\varepsilon^{(i)} = 1] = p_i $, and the result of performing $m$ trials of the algorithm \eqref{main-algo-1} are represented by the random variables $\varepsilon^{(i)}_k, \mbox{ for all } k \in \Z^+, 1 \leq k \leq m$. Consider the sample mean given by:

\begin{equation}\label{sample_mean}
X^i \,=\,  \frac{1}{m} \; \sum_{k = 1}^{m} \varepsilon^{(i)}_k.
\end{equation}

Then for a given margin of error $t > 0$
we have 
$$\Pr[ X^i - t \leq p_i \leq X^i + t ] \geq 1 - 2 \, \exp(-2 m t^2).$$
\end{theorem}

\begin{proof}
Since $\varepsilon^{(i)}$ is $0/1$ random variable, the expectation value of this variable is equal to the probability $\Pr[\varepsilon^{(i)} = 1] = p_i$. The probability inequality given in the theorem is due to the Hoeffding's inequality~\cite{Hoeffding1963}. \qed  
\end{proof}

The number of trials ($m$), required for each iteration in order to estimate $p_i$ is determined by the degree of uncertainty $\delta : \delta \in [0,1]$ and the margin of error $t$. We have
$$\Pr[ X^i - t \leq p_i \leq X^i + t ] \equiv  1 - \delta   \geq 1 - 2 \, \exp(-2 m t^2),$$
which renders $m  =  \frac{1}{2 t^2} \ln(\frac{2}{\delta})$, this is the sample complexity of the QSAC algorithm. It can be seen that the sample complexity does not scale with the input size of the algorithm. This implies that the execution time of one iteration of the algorithm is $\cO(n)$ therefore the complete algorithm ($n$ iterations) can be said to have a complexity of $\cO(n^2)$.

\subsection{Comparison with the alternatives}

Using the same analysis as above, the time complexity of an instance of the direct algorithm evaluates to $(2, 2n+4)$. Similarly the time complexity of the Forrelation algorithm is $(5, 4n+6)$, there are 5 oracle calls because each call to the derivative of the Boolean function requires two calls to the function's own oracle $U_f$. 

While evaluating the sample complexity using Hoeffding's inequality, it should be remembered that since $n$-qubit measurements are required at the end of these algorithms,  the value of $(b-a)$ from                       Proposition~\autoref{inequality} is $2^n - 1$. This implies the sample complexity for these algorithms is: $$q \,=\, \frac{2^n - 1}{2t^2} \, \ln(\frac{2}{\delta}) \,=\, \cO\left(\frac{2^n}{2t^2} \, \log(\frac{2}{\delta})\right)$$ The sample complexity in this case scales with exponentially with the input size of the algorithm. This is the primary drawback of these two alternatives when compared against the QSAC algorithm. The exponential scaling of the sample complexity is present in any algorithm that requires $n$-qubit measurements for  obtaining results.

Another important parameter to consider while comparing quantum algorithms is the number of qubits required for implementation of the circuit. This number has a practical relevance to the scalability of a given algorithm. This number is the minimum required qubits because the Boolean function oracles for functions of degree grater than $3$ require additional qubits for their implementation. It should be noted that the minimum number of qubits required for the QSAC, direct and Forrelation algorithms is $n+1$. 

The autocorrelation algorithm on the other hand requires $2n + 3$ qubits. This implies running the autocorrelation algorithm requires a quantum computer with twice as many qubits as required for implementing any other variant. The time complexity of the autocorrelation algorithm is $(2, 5n + 6)$. The sample complexity however is identical to the QSAC algorithm and doesn't scale with the input size of the function. 
  
The various complexity parameters all the algorithms discussed so far summarized in \autoref{Comparison}.

\begin{table}[htb]
\centering
\caption{\footnotesize{Table comparing the complexity values and qubit requirements between the different approaches}}
\scalebox{0.9}{
\begin{tabular}[c]{| P {10em}| P {10em}| P {10em}| P {10em}|}
\hline Algorithm variant & Query Complexity & Sample complexity $(q)$ & Qubits \\
\hline\hline Classical & $2n$ & $\frac{2^n}{t^2} \, \log(\frac{2}{\delta})$ & $-$ \\ 
\hline QSAC & $n$ & $\frac{1}{2t^2} \, \log(\frac{2}{\delta})$ & $n+1$ \\
\hline Direct \cite{Deutsch} & $2n$ & $\frac{2^n-1}{2t^2} \, \log(\frac{2}{\delta})$ & $n+1$ \\
\hline Forrelation \cite{Datta} & $5n$ & $\frac{2^n-1}{2t^2} \, \log(\frac{2}{\delta})$ & $n+1$ \\
\hline Autocorrelation \cite{BeraMT19} & $2n$ & $\frac{1}{2t^2} \, \log(\frac{2}{\delta})$ & $2n + 3$ \\
\hline
\end{tabular}}
\label{Comparison}
\end{table} 

It can be seen that the QSAC algorithm has the best values for each of these figures of merit when compared agains the other quantum algorithms presented here. It also has an exponential improvement in terms of the sampling complexity over the classical algorithm. This is speed up is a result of quantum parallelism that allows one to simultaneously query a Boolean function over all input values. Additionally, all relevant spectral information of the Boolean function was transferred to a single qubit measurement by taking advantage of quantum entanglement.

Also, these comparisons presented are only with reference to the SAC problem. The theoretical versatility of the Forrelation based algorithms  and the autocorrelation algorithm is apparent.

\section{Conclusion}

The QSAC algorithm presented here is capable of verifying SAC in Boolean functions. The algorithm contains $n$ iterations, each of which uses $\cO(n)$ gates in addition to a single call to the Boolean function oracle. The sample complexity of the algorithm is $\cO\left(\frac{1}{2t^2} \, \log(\frac{2}{\delta})\right)$ which is dependent only on the error and confidence and error bounds. From a practical standpoint, The QSAC algorithm is also the least susceptible to gate error and most capable of being adapted in to an quantum error correction framework. All of these features are a consequence of the lightweight nature of the algorithm. This however comes at a cost of versatility. Unlike the autocorrelation and forrelation based algorithms, the QSAC is designed with the sole purpose of addressing the SAC problem alone.

The complexity analysis considered here are problem independent. It can therefore be said that a lightweight quantum algorithm that solves a problem with the minimum possible number of gates, oracle calls and final measurements is likely to have properties similar to the QSAC algorithm and such algorithms are more beneficial when one considers the actual implementation on a quantum processor.

The analysis presented here has also resulted in the identification of figures of merits that can be used to compare quantum algorithms with each other.      
   
{\bf Acknowledgement:}
Research of C. A. Jothishwaran, Vishvendra Singh Poonia, and Sugata Gangopadhyay is a part of the project “Design and Development of Quantum Computing Toolkit and Capacity Building” sponsored by the Ministry of Electronics and Information Technology (MeitY) of the Government of India.

\bibliographystyle{splncs03}
%\bibliography{sugo-cryptoref}

\end{document}